\documentclass[runningheads]{llncs}
\usepackage[T1]{fontenc}
\usepackage{amsmath}
\usepackage{amsfonts}
\usepackage{dsfont}
\usepackage{ amssymb }

\usepackage{graphicx}
\usepackage{hyperref}
\usepackage{cleveref}

\usepackage{thm-restate}

\usepackage{makecell}

\usepackage{booktabs}
\usepackage{multirow}

\usepackage[table]{xcolor}

\usepackage{wrapfig}

\usepackage{subcaption}
\usepackage{tikz}

\usepackage{color}

\providecommand{\ltlN}{\operatorname{%
		\tikz[baseline]{
			\draw[line width=.12ex]
			(0,.6ex) circle (.8ex);
}}}{}

\providecommand{\ltlF}{\operatorname{%
		\tikz[baseline]{
			\draw[line width=.12ex,line join=round]
			(0ex,.6ex) -- (.95ex,1.55ex) -- (1.9ex,.6ex) -- (.95ex,-.35ex) -- cycle;
}}}{}

\providecommand{\ltlG}{\operatorname{%
		\tikz[baseline]{
			\draw[line width=.12ex,line join=round]
			(0ex,-.2ex) -- (0ex,1.3ex) -- (1.5ex,1.3ex) -- (1.5ex,.-.2ex) -- cycle;
}}}{}

\DeclareMathOperator{\ltlU}{\mathcal{U}}
\DeclareMathOperator{\ltlW}{\mathcal{W}}

\newcommand{\ap}{\mathit{AP}}
\newcommand{\pathVars}{\mathcal{V}}

\newcommand{\traceset}{\mathbb{T}}

\newcommand{\ldot}{\mathpunct{.}}

\newcommand{\quant}{\mathds{Q}}

\newcommand{\nat}{\mathbb{N}}
\newcommand{\intSet}{\mathbb{Z}}

\newcommand{\calA}{\mathcal{A}}

\newcommand{\calS}{\mathcal{S}}

\newcommand{\calI}{\mathcal{I}}

\newcommand{\frakF}{\mathfrak{F}}
\newcommand{\frakP}{\mathfrak{P}}

\newcommand{\tool}{\texttt{FOLHyper}}

\newcommand{\traceSort}{\texttt{Trace}}
\newcommand{\timeSort}{\texttt{Time}}

\newcommand{\statePred}[1]{\mathit{at}_{#1}}

\newcommand{\eahyper}{\texttt{EAHyper}}
\newcommand{\mghyper}{\texttt{MGHyper}}
\newcommand{\lmhyper}{\texttt{LMHyper}}

\newcommand\xqed[1]{%
	\leavevmode\unskip\penalty9999 \hbox{}\nobreak\hfill
	\quad\hbox{#1}}
\newcommand\demo{\xqed{$\triangle$}}

\definecolor{sat}{HTML}{117733}
\definecolor{unsat}{HTML}{882255}

\usepackage{graphicx} 
\usepackage{hyperref} 
\usepackage[firstpage]{draftwatermark} 
\usepackage{etoolbox} 
\usepackage{orcidlink}
 \usepackage[misc]{ifsym}

\newbool{artifactfunctional}
\newbool{artifactreusable}
\newbool{artifactavailable}

\setbool{artifactfunctional}{true}
\setbool{artifactreusable}{true}
\setbool{artifactavailable}{true}
\newcommand{\artifacturl}{https://doi.org/10.5281/zenodo.11158978}
\newcommand{\artifactvposition}{16.5cm}

\renewcommand{\artifactvposition}{22cm}
\SetWatermarkAngle{0}
\SetWatermarkText{\raisebox{\artifactvposition}{%
		\hspace{5.8cm}%
		\href{\artifacturl}{\includegraphics[width=1.7cm]{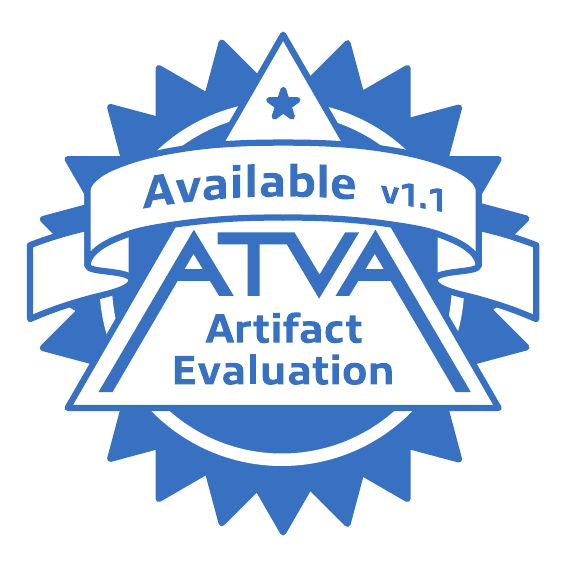}}
		\hspace{-8mm}
		\includegraphics[width=1.7cm]{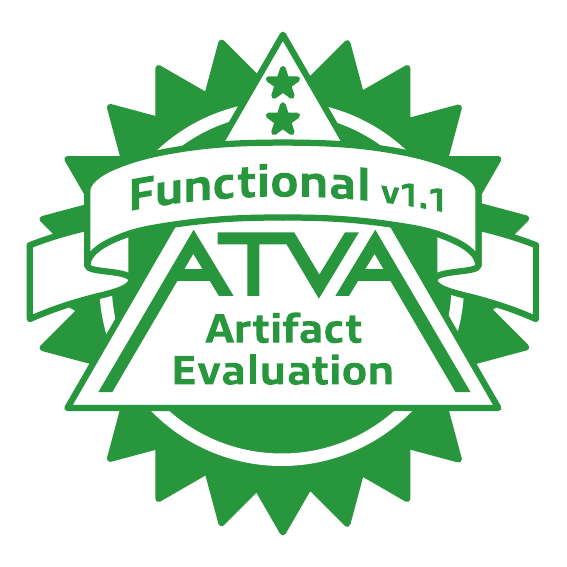}
		\hspace{-8mm}
		\includegraphics[width=1.7cm]{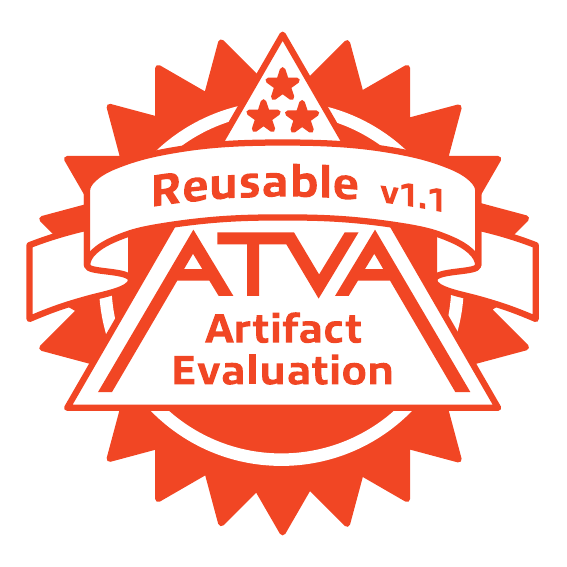}
}}

\newif\iffullversion
\fullversiontrue

\newcommand{\ifFull}[2]{\iffullversion#1\else#2\fi}

\begin{document}
\title{Checking Satisfiability of Hyperproperties \\using First-Order Logic}

\author{Raven Beutner$^{\text{(\Letter)}}$\orcidlink{0000-0001-6234-5651} \and Bernd Finkbeiner\orcidlink{0000-0002-4280-8441}}
\authorrunning{R.~Beutner and B.~Finkbeiner}
%
\institute{CISPA Helmholtz Center for Information Security, Saarbrücken, Germany\\
	\email{\{raven.beutner,finkbeiner\}@cispa.de}
} 

\maketitle              
\begin{abstract}
Hyperproperties are system properties that relate multiple execution traces and occur, e.g.,  when specifying security and information-flow properties. 
Checking if a hyperproperty is satisfiable has many important applications, such as testing if some security property is contradictory, or analyzing implications and equivalences between information-flow policies. 
In this paper, we present \tool{}, a tool that can automatically check satisfiability of hyperproperties specified in the temporal logic HyperLTL.
\tool{} reduces the problem to an equisatisfiable first-order logic (FOL) formula, which allows us to leverage FOL solvers for the analysis of hyperproperties.
As such, \tool{} is applicable to many formulas beyond the decidable $\exists^*\forall^*$ fragment of HyperLTL. 
Our experiments show that \tool{} is particularly useful for proving that a formula is unsatisfiable, and complements existing bounded approaches to satisfiability.
\end{abstract}

\section{Introduction}

Hyperproperties \cite{ClarksonS08} are system properties that relate multiple execution traces of a system and commonly occur in areas such as information-flow control \cite{Rabe16,ZdancewicM03}, robustness \cite{BiewerDFGHHM22}, planning \cite{WangNP20}, linearizability \cite{HsuSB21,HerlihyW90}, and opacity \cite{AnandMTZ24}.
A prominent logic to express hyperproperties is HyperLTL, which extends linear-time temporal logic (LTL) with explicit quantification over traces \cite{ClarksonFKMRS14}.
For instance, HyperLTL can express generalized non-interference (GNI) \cite{McCullough88}, stating that the high-security input (modeled via atomic proposition $h$) of a system does not influence the observable output, as follows:
\begin{align*}\label{eq:gni}
	\forall \pi_1\ldot \forall \pi_2\ldot \exists \pi_3\ldot \ltlG \big((l_{\pi_1} \leftrightarrow l_{\pi_3} )\land (o_{\pi_1} \leftrightarrow o_{\pi_3}) \big) \land \,\ltlG  (h_{\pi_2} \leftrightarrow h_{\pi_3}) \tag{GNI}
\end{align*}
This formula requires that for any pair of traces $\pi_1, \pi_2$, there exists a third trace $\pi_3$ that agrees with the low-security observations (i.e., low-security input $l$ and public output $o$) of $\pi_1$, and with the high-security input $h$ of $\pi_2$. 
Phrased differently, any low-security observation of an attacker is possible under \emph{all} possible high-security input sequences.
Similarly, we can express a variant of McLean's non-inference (NI) \cite{McLean94} as follows:
\begin{align*}\label{eq:ni}
	\forall \pi_1\ldot \exists \pi_2\ldot \ltlG \big((l_{\pi_1} \leftrightarrow l_{\pi_2} )\land (o_{\pi_1} \leftrightarrow o_{\pi_2})\big) \land \, \ltlG ( \neg h_{\pi_2} )\tag{NI}
\end{align*}
Stating that for every $\pi_1$, some trace $\pi_2$ has the same low-security observations despite having a fixed high-security input (we require that $h$ never holds on $\pi_2$).

\paragraph{Satisfiability.}

We are interested in the question of whether a HyperLTL formula can be satisfied by some set of traces. 
Such satisfiability checks can spot inconsistencies in system requirements, e.g., information-flow requirements that are incompatible with functional specifications.  
Moreover, implication (and thus equivalence) checks between HyperLTL formulas reduce to satisfiability: to analyze if a HyperLTL formula $\varphi_A$ implies $\varphi_B$, we check if some set of traces can disprove the implication, i.e., satisfies $\varphi_A \land \neg \varphi_B$. 
As hyperproperties occur in many areas within computer science, satisfiability and implication checks for them are fundamental problems in many applications; for example, in optimized monitoring of fairness and symmetry properties \cite[\S 4.1]{FinkbeinerHST19}.

\paragraph{\tool{}.}

In this paper, we present \tool{}, an automatic satisfiability checker for HyperLTL. 
At its core, \tool{} encodes the HyperLTL formula into an equisatisfiable first-order logic (FOL) formula.
In case the HyperLTL formula is temporally safe -- i.e., the quantifier-free LTL body denotes a safety property, as in \ref{eq:gni} and  \ref{eq:ni} -- \tool{} produces a pure FOL formula.
For this, we employ a similar encoding to the one presented (but not implemented) in \cite{BeutnerCFHK22}.
For general HyperLTL, we present a novel encoding into FOL modulo theories. 
\tool{} thus allows us to leverage off-the-shelf FOL/SMT solvers and facilitates the use of common first-order techniques such as resolution \cite{DavisP60} and tableaux \cite{beth1955semantic} in the realm of hyperproperties. 
We evaluate \tool{} on a range of benchmarks and compare it with the existing HyperLTL satisfiability checkers \eahyper{} \cite{FinkbeinerHS17}, \mghyper{} \cite{FinkbeinerHH18}, and \lmhyper{} \cite{BeutnerCFHK22}. 
Our results show \tool{}'s FOL-based encoding works well on many formulas that could already be handled using existing methods, and can tackle important instances out of reach of previous tools. 
For example, \tool{} is the first tool that can efficiently show satisfiability for formulas that require a non-singleton model. 
Most notably, \tool{} is also the first tool that can show \emph{unsatisfiability} of HyperLTL formulas with arbitrary quantifier prefixes and thus \emph{prove} implications involving complex (alternating) formulas such as \ref{eq:gni} and  \ref{eq:ni}.

\section{Related Work.}\label{sec:related-work}

The satisfiability problem of hyperproperties -- and in particular of HyperLTL -- has been studied extensively \cite{FinkbeinerH16,FortinKT021,Mascle020,BeutnerCFHK22,CoenenFHH19,GutsfeldMO20}. 
The $\exists^*\forall^*$ fragment (containing all formulas where no existential quantifier occurs in the scope of a universal quantifier) is the largest fragment (defined by quantifier structure) for which satisfiability is decidable \cite{FinkbeinerH16}. 
The addition of a $\forall\exists$ alternation leads to undecidability ($\Sigma_1^1$-completeness) \cite{FortinKT021}.
Decidable fragments can be obtained by bounding the size of a model \cite{Mascle020} or by (severely) restricting the use of temporal operators \cite{Mascle020,BeutnerCFHK22}.
By extension of Kamp's theorem \cite{kamp1968tense}, each HyperLTL formula can be translated to an equivalent formula in FOL of order (FO$[<]$)  using a special equal-level predicate (FO$[<,E]$) \cite{Finkbeiner017,CoenenFHH19}; a fragment that existing solvers do not support.
In contrast, \tool{} produces an equisatisfiable (but not necessarily equivalent) encoding in either pure FOL or FOL modulo linear integer arithmetic (LIA), which allows us to employ off-the-shelf FOL and SMT solvers.

The theoretical studies into HyperLTL satisfiability have been complemented by extensive tool development:
\eahyper{} \cite{FinkbeinerHS17} is applicable to the decidable $\exists^*\forall^*$ fragment and reduces to LTL satisfiability by  unfolding all dependencies between existential and universal quantifiers;
\mghyper{} \cite{FinkbeinerHH18} attempts to find finite models of bounded size via a QBF encoding; and \lmhyper{} \cite{BeutnerCFHK22} is applicable to $\forall^1\exists^*$ formulas (i.e., formulas with a single leading universal quantifier) and attempts to find the \emph{largest} model using iterative automaton projections.    

\section{Preliminaries and HyperLTL}

\paragraph{HyperLTL.}
We assume that $\ap$ is a finite set of atomic propositions (AP). 
As the basic specification language for hyperproperties, we use HyperLTL, an extension of LTL with explicit quantification over traces \cite{ClarksonFKMRS14}.
Let $\pathVars = \{\pi_1, \pi_2, \ldots\}$ be a set of \emph{trace variables}.
HyperLTL formulas are generated by the following grammar
\begin{align*}
	\psi &:= a_\pi \mid \psi \land \psi \mid \neg \psi \mid \ltlN \psi \mid \psi \ltlU \psi  \\
	\varphi &:= \forall \pi \ldot \varphi \mid \exists \pi \ldot \varphi \mid \psi
\end{align*}
where $a \in \ap$ is an atomic proposition, $\pi \in \pathVars$ is a trace variable, and $\ltlN, \ltlU$ denote LTL's \emph{next} and \emph{until} operator, respectively. 
We only consider closed formulas where for each atom $a_\pi$, trace variable $\pi$ is bound by some outer quantifier. 
We use the usual derived boolean constants and connectives $\mathit{true}, \mathit{false}, \lor, \to, \leftrightarrow$ and temporal operators \emph{eventually} $\ltlF \psi := \mathit{true} \ltlU \psi$, \emph{globally} $\ltlG \psi := \neg \ltlF \neg \psi$, and \emph{weak until} $\psi_1 \ltlW \psi_2 := (\psi_1 \ltlU \psi_2) \lor \ltlG \psi_1$.
Given a set of traces $\traceset \subseteq (2^\ap)^\omega$, we can define the semantics of a HyperLTL formula $\varphi$ as expected: Quantification ranges over traces in $\traceset$ and binds them to the respective trace variable;
The LTL body $\psi$ is evaluated synchronously on the resulting traces, where atomic formula $a_\pi$ holds whenever $a$ holds in the current position on the trace bound to $\pi$; Boolean and temporal operators are evaluated as expected. 
We write $\traceset \models \varphi$ if $\traceset$ satisfies $\varphi$.
We give the full semantics in \ifFull{\Cref{app:semantics}}{the full version}.
A formula $\varphi$ is \emph{satisfiable} if there exists some set of traces $\traceset \subseteq (2^\ap)^\omega$ with $\traceset \neq \emptyset$ and $\traceset \models \varphi$.

\paragraph{NSA.}

A \emph{nondeterministic safety automaton} (NSA) over some finite alphabet $\Sigma$ is a tuple $\calA = (Q, Q_0, \delta, B)$ where $Q$ is a finite set of states, $Q_0 \subseteq Q$ is a set of initial states, $\delta \subseteq Q \times \Sigma \times Q$ is a transition relation, and $B \subseteq Q$ is a set of bad states. 
A run of $\calA$ on a word $u \in \Sigma^\omega$ is an infinite sequence $\rho \in Q^\omega$ such that $\rho(0) \in Q_0$ and for every $i \in \nat$, $(\rho(i), u(i), \rho(i+1)) \in \delta$. 
The run $\rho$ is accepting if it \emph{never} visits a state in $B$.
We say an LTL formula $\psi$ denotes a \emph{safety} property if there exists an NSA that accepts exactly the words that satisfy $\psi$ \cite{KupfermanV99}.

\paragraph{First-Order Logic.}

A \emph{many-sorted signature} is a tuple $(\calS, \frakF, \frakP)$ where $\calS$ is a set of sorts, $\frakF$ is a set of function symbols, and $\frakP$ is a set of predicate symbols.
Each $f \in \frakF$ has a unique sort $f : S_1 \times \cdots \times S_m \to S$ (where $S_1, \ldots, S_m, S \in \calS$).
Likewise, each predicate $P \in \frakP$ has a unique sort $P : S_1 \times \cdots \times S_m$.
A \emph{term of sort $S$} is recursively defined as either \textbf{(1)} a variable $x$ of sort $S$, or \textbf{(2)} a function application $f(t_1, \ldots, t_m)$ where $f \in \frakF$ is a function symbol of sort $f : S_1 \times \cdots \times S_m \to S$, and $t_1, \ldots, t_m$ are terms of sort $S_1, \ldots, S_m$, respectively. 
A \emph{FOL formula} is recursively defined as 
\begin{align*}
	\theta := \neg \theta \mid \theta \land \theta \mid \forall x : S \ldot \theta \mid \exists x : S \ldot \theta \mid P(t_1, \ldots, t_m)
\end{align*}
where $S \in \calS$ is a sort, $x$ is a variable of sort $S$, $P : S_1 \times \cdots \times S_m$ is a predicate symbol, and $t_1, \ldots, t_m$ are terms of sort $S_1, \ldots, S_m$, respectively. 
A \emph{FOL interpretation} of signature $(\calS, \frakF, \frakP)$ is a mapping $\calI$ that maps each sort $S \in \calS$ to a set $S^\calI$, each function symbol $f : S_1 \times \cdots \times S_m \to S$ to a function $f^\calI : S_1^\calI \times \cdots \times S_m^\calI\to S^\calI$, and each predicate symbol $P : S_1 \times \cdots \times S_m$ to a relation $P^\calI \subseteq S_1^\calI \times \cdots \times S_m^\calI$.
The interpretation $\calI$ satisfies a formula $\theta$, written $\calI \models \theta$,  if $\theta$ evaluates to $\mathit{true}$ under $\calI$ (defined as expected). We refer the reader to \cite{barwise1977introduction} for details.

\section{First-Order Logic Encoding}

We consider a fixed HyperLTL formula $\varphi = \quant_1 \pi_1\ldots \quant_n \pi_n\ldot \psi$, where $\quant_1, \ldots, \quant_n \in \{\forall, \exists\}$ are quantifiers and $\psi$ is the quantifier-free LTL body.

\subsection{Temporally Safe HyperLTL}\label{sec:safe}

We first consider the fragment of HyperLTL that we can encode in pure FOL.
We say a $\varphi$ is \emph{temporally safe} if the LTL body $\psi$ denotes a safety property \cite{BeutnerCFHK22}. 
That is, there exist a NSA $\calA_\psi = (Q_\psi, Q_{0, \psi}, \delta_\psi, B_\psi)$ over alphabet $2^{\ap \times \{\pi_1, \ldots, \pi_n\}}$ that accepts exactly the words $u \in \big(2^{\ap \times \{\pi_1, \ldots, \pi_n\}}\big)^\omega$ that satisfy $\psi$ (recall that atomic formulas in $\psi$ have the form $a_{\pi_i} \in \ap \times \{\pi_1, \ldots, \pi_n\}$).
Note that temporal safety is a decidable criterium \cite{KupfermanV99} that only refers to the LTL body of $\varphi$; it does not refer to the models of $\varphi$, as, e.g., used in the definition of \emph{hypersafety} \cite{ClarksonS08}.

Intuitively, as $\psi$ denotes a safety property, we do not need an accurate interpretation of time (as, e.g., used in the HyperLTL semantics).
Instead, it suffices to be able to -- for each time point -- refer to the \emph{successor} time point, something we can easily achieve in  pure FOL (cf.~\cite{BeutnerCFHK22}).
We use a signature $(\calS, \frakF, \frakP)$ where \textbf{(1)} $\calS := \{\traceSort{}, \timeSort{}\}$, \textbf{(2)} $\frakF$ contains constants (i.e., nullary functions) $i_0 : \timeSort{}$ and $t_0 : \traceSort{}$ (by adding the constant $t_0$, we ensure that the sort $\traceSort$ is non-empty), and a successor function $\mathit{succ} : \timeSort{} \to \timeSort{}$, and \textbf{(3)} $\frakP$ includes the following: for each $a \in \ap$ a predicate $P_a :  \traceSort{} \times \timeSort{}$, and for each $q \in Q_\psi$ a predicate $\statePred{q} : (\times_{i=1}^n \traceSort{})\times \timeSort{}$.

We define $\theta_\varphi$ as the FOL formula over $(\calS, \frakF, \frakP)$ in \Cref{fig:encoding}.
We first mimic the quantification in $\varphi$ over variables $x_1, \ldots, x_n$ of sort $\traceSort$ (\ref{eq:quant}). 
Intuitively, predicate $P_a(x, i)$ holds whenever the AP $a$ holds on trace $x$ in timestep $i$.
Consequently, each first-order interpretation will (implicitly) assign a concrete trace to each element of sort $\traceSort$.
Our encoding ensures that the traces assigned to $x_1, \ldots, x_n$ satisfy $\psi$, i.e., $\calA_\psi$ has an accepting run. 
To track the run of $\calA_\psi$, we use the $\{\statePred{q}\}_{q \in Q_\psi}$ predicates:
Intuitively, $\statePred{q}(x_1, \ldots, x_n, i)$ holds, whenever $\calA_\psi$ -- when reading traces $x_1, \ldots, x_n$ -- \emph{can} be in state $q$ at time point $i$.
We impose three constraints that encode a valid run of $\calA_\psi$.
First, we require that $\calA_\psi$ starts in some initial state at our fixed initial time constant $i_0$ (\ref{eq:safety-quant}).
Second, the timestep should encode valid transitions of $\calA_\psi$. 
To express this, we consider all possible timepoints $i : \timeSort$ and states $q \in Q_\psi$, assuming that the automaton can be in state $q$ at timepoint $i$ (i.e., $\statePred{q}(x_1, \ldots, x_n, i)$) (\ref{eq:safety-premise}).
In this case, there should exist some edge $(q, \sigma, q') \in \delta_\psi$ (starting from the current state $q$), such that we \textbf{(1)} move to $q'$ at timepoint $\mathit{succ}(i)$ (i.e., $\statePred{q'}(x_1, \ldots, x_n, \mathit{succ}(i))$ holds; the first conjunct in \ref{eq:safety-cons}), and \textbf{(2)} $\sigma \in 2^{\ap \times \{\pi_1, \ldots, \pi_n\}}$ is the label at the current time point $i$.
The latter means that for every $a_{\pi_j} \in \ap \times \{\pi_1, \ldots, \pi_n\}$, we have $a_{\pi_j} \in \sigma$ if and only if $a$ holds at time point $i$ on trace $x_j$.
Phrased differently, $P_a(x_j, i)$ holds for all $a_{\pi_j} \in \sigma$, and $\neg P_a(x_j, i)$ holds for all $a_{\pi_j} \not\in \sigma$ (the last two conjuncts in \ref{eq:safety-cons}).
Lastly, $\calA_\psi$ should never visit a bad state (\ref{eq:safety-safe}).

\begin{figure}[!t]
	\begin{align}
		&\quant_1 \, x_1 : \traceSort{} \ldots \quant_n \,  x_n  : \traceSort{} \ldot  \label{eq:quant} \\
		&\quad\Big(\bigvee_{q \in Q_{0, \psi}} \statePred{q}\big(x_1, \ldots, x_n, i_0\big) \Big) \, \land \label{eq:safety-quant} \\
		&\quad\Bigg(\forall i : \timeSort{}\ldot \bigwedge_{q \in Q_\psi}  \bigg[ \statePred{q}\big(x_1, \ldots, x_n, i\big) \to  \label{eq:safety-premise} \\
		&\quad\quad\bigvee_{(q, \sigma, q') \in \delta_\psi} \!\!\bigg(  \statePred{q'}\big(x_1, \ldots, x_n, \mathit{succ}(i)\big) \land \bigwedge_{a_{\pi_j} \in \sigma} \! P_a(x_j, i) \land \!\!\bigwedge_{a_{\pi_j} \not\in \sigma} \!\neg P_a(x_j, i) \bigg) \bigg] \Bigg) \, \land  \label{eq:safety-cons} \\
		&\quad\Big(\forall i : \timeSort{}\ldot \bigwedge_{q \in B_\psi} \neg \statePred{q}\big(x_1, \ldots, x_n, i\big)\Big) \label{eq:safety-safe}
	\end{align}
	\vspace{-4mm}
	\caption{We depict the FOL formula $\theta_\varphi$ over signature $(\calS, \frakF, \frakP)$.}\label{fig:encoding}
\end{figure}

\begin{restatable}{theorem}{correct}\label{theo:correct1}
	Assume $\varphi$ is a temporally safe HyperLTL formula.
	Then $\varphi$ is satisfiable if and only if there exists a FOL interpretation that satisfies $\theta_\varphi$.
\end{restatable}

\begin{remark}\label{rem:old}
	Compared to \cite{BeutnerCFHK22}, we modify the FOL encoding by using a successor \emph{function} instead of a successor \emph{predicate}, requiring a modified correctness proof (see \ifFull{\Cref{app:proof}}{the full version}).
	Our experiments with FOL solvers show that our function-based encoding is superior to a predicate-based encoding. 
	\demo
\end{remark}

\paragraph{Finite Models.}

Most existing HyperLTL satisfiability checkers search for finite models:
\mghyper{} \cite{FinkbeinerHH18} uses an \emph{explicit} bound on the size of the model that is iteratively increased. 
\eahyper{} \cite{FinkbeinerHS17} uses an \emph{implicit} bound, i.e., a $\exists^n\forall^m$ formula never requires a model with more than $n$ traces \cite{FinkbeinerH16}.
We can analyze our FOL encoding to understand the theoretical relation between \tool{} and \mghyper{}:

\begin{restatable}{theorem}{finite}\label{prop:finite}
	Assume $\varphi$ is a temporally safe HyperLTL formula.
	Then, $\varphi$ is satisfiable by a \emph{finite} set of traces (i.e., there exists a set of traces $\traceset   \subseteq 2^\ap$ with $\traceset \neq \emptyset$, $|\traceset| < \infty$, and $\traceset \models \varphi$) if and only if $\theta_\varphi$ is satisfiable by a \emph{finite} FOL interpretation (i.e., there exists an FOL interpretation $\calI$ with $\calI \models \theta_\varphi$ and $|S^\calI| < \infty$ for every sort $S \in \calS$).
\end{restatable}
\begin{proof}[Sketch]
	Computing a finite model of $\varphi$ from a finite FOL interpretation of $\theta_\varphi$ is straightforward.
	For the reverse direction, we show that we can restrict to lasso-shaped traces and lasso-shaped runs in $\calA_\psi$.
	This allows us to use a \emph{cyclic} (and thus finite) interpretation of sort $\timeSort$, based on the least common multiple of all lasso lengths. See \ifFull{\Cref{app:finite}}{the full version} for details. \qed
\end{proof}

\begin{remark}\label{rem:finite}
	By \Cref{prop:finite}, any FOL solver that searches for finite models of (increasing) bounded size will terminate in at least those instances that \mghyper{} -- the only other HyperLTL satisfiability solver applicable to arbitrary quantifier prefixes -- can solve.  
	Note that deciding if a FOL formula has a finite interpretation of bounded size is decidable \cite{mccune1994davis,ReynoldsTGK13}, and practical SEM-style \cite{ZhangZ95} or MACE-style \cite{mccune1994davis,JanotaS18} finite model finders exist; e.g., \texttt{Paradox} \cite{claessen2003new}. \demo
\end{remark}

\subsection{Full HyperLTL}\label{sec:full}

The above encoding is  limited to temporally safe HyperLTL.
We can extend the encoding to work for full HyperLTL by employing a fixed background theory such as linear integer arithmetic (LIA).
The fixed background theory admits a precise interpretation of discrete timesteps, which allows us to generalize our construction from safety automata to Büchi automata, thus extending our encoding to HyperLTL formulas with arbitrary LTL body. 
See \ifFull{\Cref{app:full}}{the full version} for details.
We note that our extended LIA encoding is, currently, mostly of theoretical interest:
Our experiments with \tool{} show that the addition of LIA often yields instances that current SMT solvers cannot solve.

\section{Implementation}

We have implemented the encodings from \Cref{sec:safe,sec:full} in an \texttt{F\#} tool we call \tool{}.
For the translation of the LTL body to an automaton, we use \texttt{spot} \cite{Duret-LutzRCRAS22}.
\tool{} features a minimalist command-line interface and reads a HyperLTL formula using an extension of \texttt{spot}'s default LTL format. 
Depending on the outcome of the underlying FOL query, \tool{} outputs \texttt{SAT}, \texttt{UNSAT}, or \texttt{UNKNOWN} (or diverges if the underlying FOL solver diverges).
The user has multiple command-line options that guide \tool{}'s behavior.
First, the user can choose between the function-based encoding (presented in \Cref{sec:safe}), the predicate-based encoding (cf.~\Cref{rem:old}), and the encoding into FOL modulo LIA (\Cref{sec:full}). 
Second, the user can determine if \tool{} should produce the FOL formula in the \texttt{SMTLIB} format \cite{barrett2010smt} (supported by most SMT solvers) or the \texttt{TPTP} format \cite{SutcliffeSY94} (the format used for the annual ATP competition, supported by most pure FOL solvers). 
Lastly, the user can determine if \tool{} should only compute the encoding or also attempt to solve it using existing FOL/SMT solvers such as \texttt{Vampire} \cite{KovacsV13}, \texttt{EProver} \cite{SchulzCV19}, \texttt{iProver} \cite{Korovin08}, \texttt{Paradox} \cite{claessen2003new}, \texttt{z3} \cite{MouraB08}, and \texttt{cvc5} \cite{BarbosaBBKLMMMN22}. 

The experiments in this paper were conducted on a MacBook with an M1 Pro CPU and 32 GB of memory. 
If not specified otherwise, we use \texttt{Vampire} (version 4.6.1) \cite{KovacsV13} as \tool{}'s default FOL solver.

\begin{table}[!t]
	
	\caption{We check if $\mathit{QN}(\mathit{n})$ implies $\mathit{QN}(\mathit{m})$ where $n$ is given by the row and $m$ by the column. For each combination, we give the time of \eahyper{} \cite{FinkbeinerHS17} (using \texttt{PLTL} \cite{Schwendimann98} for LTL satisfiability checks) and \tool{} in seconds. The timeout (denoted ``-'') is set to 60 seconds.
	}\label{tab:eval-qn}
	
	\centering
	
	\vspace{1mm}
	
	\small
	
	\def\arraystretch{1.1}
	\setlength\tabcolsep{1.7mm}

	\begin{tabular}{l@{\hspace{4mm}}lccccccc}
		\toprule
		& & 1 & 2 & 3 & 4 & 5 & 6 & 7 \\
		\midrule
		\multirow{2}{*}{1} & \eahyper{} \cite{FinkbeinerHS17} &  \textbf{0.01} &  \textbf{0.01 }&  \textbf{0.01 }&  \textbf{0.01} &  \textbf{0.01} &  \textbf{0.01} &  \textbf{0.01 }\\
		& \tool{} &  0.17  &  0.17  &  0.18  &  0.17  &  0.17  &  0.17  &  0.16 \\
		\midrule
		\multirow{2}{*}{2}& \eahyper{} \cite{FinkbeinerHS17}&  \textbf{0.01} &  \textbf{0.01 }&  \textbf{0.01} &  \textbf{0.01} &  \textbf{0.02} &  \textbf{0.03} &  \textbf{0.04} \\
		& \tool{} &  0.18  &  0.19  &  0.17  &  0.17  &  0.17  &  0.16  &  0.17 \\
		\midrule
		\multirow{2}{*}{3}&\eahyper{} \cite{FinkbeinerHS17}&  \textbf{0.01} &  \textbf{0.01} &  \textbf{0.02} &  \textbf{0.03} &  \textbf{0.06} &  \textbf{0.13 }&  0.27 \\
		&\tool{}&  0.17  &  0.16  &  0.17  &  0.17  &  0.17  &  0.17  &  \textbf{0.18} \\
		\midrule
		\multirow{2}{*}{4}&\eahyper{} \cite{FinkbeinerHS17}&  \textbf{0.01} &  \textbf{0.01} &  \textbf{0.06} & \textbf{ 0.13} &  0.36 &  0.95 &  2.19 \\
		&\tool{}&  0.18  &  0.17  &  0.17  &  0.17  &  \textbf{0.17}  &  \textbf{0.17}  &  \textbf{0.18} \\
		\midrule
		\multirow{2}{*}{5}&\eahyper{} \cite{FinkbeinerHS17}&  \textbf{0.01} &  \textbf{0.02 }&  1.06 & 15.12 &  6.26 & 11.98 & 26.86 \\
		&\tool{}&  0.17  &  0.17  &  \textbf{0.17}  &  \textbf{0.18}  &  \textbf{0.18}  &  \textbf{0.18}  &  \textbf{0.17}\\
		\midrule
		\multirow{2}{*}{6}&\eahyper{} \cite{FinkbeinerHS17}&  \textbf{0.01} &  \textbf{0.05} & 25.70 & - & - & - & - \\
		&\tool{}&  0.20  &  0.19  &  \textbf{0.18}  &  \textbf{0.19}  &  \textbf{0.17}  &  \textbf{0.19 } &  \textbf{0.17} \\
		\midrule
		\multirow{2}{*}{7}&\eahyper{} \cite{FinkbeinerHS17}& \textbf{0.01} &  \textbf{0.11} & - & - & - & - & - \\
		&\tool{}& 0.18  &  0.19  &  \textbf{0.19}  &  \textbf{0.19} &  \textbf{0.21}  &  \textbf{0.20}  &  \textbf{0.20} \\
		\bottomrule
\end{tabular}

\end{table}

\section{Evaluation on $\exists^*\forall^*$ Formulas}\label{sec:ef}

The $\exists^*\forall^*$ fragment of HyperLTL constitutes the largest fragment where satisfiability is decidable \cite{FinkbeinerH16}.
\eahyper{} \cite{FinkbeinerHS17} implements a decision procedure for $\exists^n\forall^m$ formulas by explicitly unfolding all ($n^m$ many) dependencies between traces, thereby producing an equisatisfiable LTL formula.
We compare \tool{} and \eahyper{} on a subset of the benchmarks used to evaluate \eahyper{} in \cite{FinkbeinerHS17}.

\paragraph{Quantitative Noninterference.}

Quantitative noninterference (QN) requires that the leakage of a system is bound by some constant $c$ \cite{Smith09,ClarkHM05}.
In HyperLTL, we can express a variant of QN by requiring that for any input ($\mathit{In} \subseteq \ap$), there are at most $c \in \nat$ many possible outputs ($\mathit{Out} \subseteq \ap$) (see \cite{FinkbeinerHS17} for details):
\begin{align*}
	\mathit{QN}(c) := \forall \pi_0 \ldots \forall \pi_c \ldot \neg \bigg( \Big(\bigwedge_{i} \bigwedge_{a \in \mathit{In}} (a_{\pi_i} \leftrightarrow a_{\pi_0}) \Big) \land \Big(\bigwedge_{i \neq j} \; \bigvee_{a \in \mathit{Out}} (a_{\pi_i} \not\leftrightarrow a_{\pi_j}) \Big) \bigg)
\end{align*}
For varying values of $n, m \in \nat$, we check whether $\mathit{QN}(n)$ implies $\mathit{QN}(m)$.
We depict the runtimes in \Cref{tab:eval-qn}. 
For small values of $n$ and $m$, the unfolded LTL formula produced by \eahyper{} is very small, and LTL-specific satisfiability tools are more performant than our FOL encoding. 
However, for larger values, the size of the LTL formulas constructed by \eahyper{} grows exponentially, and so does the overall runtime. 
In contrast, \tool{} can check such formulas very effectively as the FOL encoding is of polynomial size.

\paragraph{Random Formulas.}

Following \cite{FinkbeinerHS17}, we also compare \tool{} and \eahyper{} on random $\exists^*\forall^*$ formulas. 
We, again, observe that \eahyper{} is faster than \tool{} on instances with few quantifiers. 
However, when the number of quantifiers in the formula increases, the unfolding performed by \eahyper{} is clearly outperformed by \tool{}.
We give a table in \ifFull{\Cref{app:eahyper}}{the full version}.
We conclude that \tool{} scales better on $\exists^*\forall^*$ formulas when the number of quantifiers increases. 

\section{Evaluation on $\forall^1\exists^*$ Formulas}

In the previous section, we focused on formulas for which satisfiability is decidable. 
We now consider formulas beyond the $\exists^*\forall^*$ fragment, moving us to the realm of undecideability. 
We compare with \mghyper{} \cite{FinkbeinerHH18} and \lmhyper{} \cite{BeutnerCFHK22}.

\paragraph{Random Formulas.}

Following \cite{BeutnerCFHK22,FinkbeinerHH18}, we sample random temporally safe $\forall^1\exists^m$ formulas for varying values of $m$.
For each $m$, we sample 20 random formulas and report the results in \Cref{tab:random_fe}.
As stated in \Cref{prop:finite}, \tool{} combined with \texttt{Paradox} can, in theory, solve at least all instances solved by \mghyper{}.
As already noted in \cite{BeutnerCFHK22}, random generation often yields formulas that are satisfiable by a single trace model, as atomic propositions are seldom shared across different trace variables. 
This explains the high success rate (and low computation time) of \mghyper{}, even though \mghyper{} can \emph{never} prove unsatisfiability.
When comparing \lmhyper{} and \tool{}, we observe that larger values for $m$ increase the runtime of \lmhyper{} slightly (which performs automaton complementations for each trace variable) but do not seem to impact \tool{}'s performance.

\begin{table}[!t]
	\caption{We evaluate on random temporally safe $\forall^1\exists^m$ for varying values of $m$ (the column). We give the percentage of solved formulas (Sol) within a timeout of 10 seconds, and the average runtime (in seconds) on solved instances ($t$). }\label{tab:random_fe}
	\vspace{1mm}
	\centering
	\small
	
	\def\arraystretch{1.1}
	\setlength\tabcolsep{1.7mm}
	
	\begin{tabular}{lcccccccc}
		\toprule
		& \multicolumn{2}{c}{\textbf{1}} & \multicolumn{2}{c}{\textbf{3}} & \multicolumn{2}{c}{\textbf{5}} & \multicolumn{2}{c}{\textbf{7}} \\
		\cmidrule[1pt](lr){2-3}
		\cmidrule[1pt](lr){4-5}
		\cmidrule[1pt](lr){6-7}
		\cmidrule[1pt](l){8-9}
		& \textbf{Sol} & $\boldsymbol{t}$ & \textbf{Sol} & $\boldsymbol{t}$ & \textbf{Sol} & $\boldsymbol{t}$ & \textbf{Sol} & $\boldsymbol{t}$ \\
		\midrule
		\mghyper{} \cite{FinkbeinerHH18} & 95 \% & 0.01 & 95 \% & 0.01 & 100 \% & 0.01 & 100 \% & 0.02 \\
		\lmhyper{} \cite{BeutnerCFHK22} & 85 \% & 0.22 & 95 \% & 0.29 & 95 \% & 0.34 & 90 \% & 0.39  \\
		 \tool{} (\texttt{Vampire}) & 60 \%& 0.18 & 45 \% & 0.21 & 70 \% & 0.26 & 70 \% & 0.18 \\
		 \tool{} (\texttt{Paradox}) & 95 \%& 0.19 & 100 \% & 0.21 & 100 \% & 0.20 & 100 \% & 0.19 \\
		\bottomrule
	\end{tabular}

\end{table}

\paragraph{Enforcing Non-Singleton Models.}

In \Cref{tab:random_fe}, \mghyper{}'s bounded approach performs very well, as most random formulas are satisfiable by a model with a single trace.
To test this, we specifically design formulas that pose a bound on the size of the smallest model. 
Concretely, given $n, b \in \nat$, we define 
\begin{align*}
	 \textstyle \mathit{enforceModel}(n, b) := \exists \pi_1 \ldots \pi_n \ldot \bigwedge_{i \neq j} \ltlF^{\leq b} (a_{\pi_i} \not\leftrightarrow a_{\pi_j}),
\end{align*}
where we write $\ltlF^{\leq b}$ for the bounded eventually (which can be expressed using $\ltlN$s).
The formula requires at least $n$ traces, which (pairwise) differ on AP $a$ within the first $b$ steps; any model must thus have at least $n$ distinct traces.
Note that the above formula is satisfiable iff $n \leq 2^b$.
We check $\mathit{enforceModel}(n, b)$ for varying $n, b$ in \Cref{tab:enforce}. 
We observe that \mghyper{} can show that there exists a model with a single trace ($n = 1$), but already times out when we require at least two traces ($n \geq 2$). 

\begin{table}[!t]
	\caption{We check $\mathit{enforceModel}(n, b)$ where $n$ is given by the column and $b$ by the row.
		We display the runtime in \textcolor{sat}{\textbf{green}} if the formula is satisfiable, and \textcolor{unsat}{\textbf{red}} if it is unsatisfiable. The timeout (denoted ``-'') is set to 60 seconds.
	}\label{tab:enforce}
	
	\vspace{1mm}
	\centering
	
	\def\arraystretch{1.1}
	\setlength\tabcolsep{1.7mm}
		
		\small
		\begin{tabular}{l@{\hspace{4mm}}lccccc}
			\toprule
			& & \textbf{1} & \textbf{2} & \textbf{3} & \textbf{4} & \textbf{5} \\
			\toprule
			\multirow{2}{*}{$\boldsymbol{b = 1}$} & \mghyper{} \cite{FinkbeinerHH18} & \color{sat}\textbf{0.02}& - & - & - & -\\
			& \tool{} & \color{sat} 0.17 & \color{sat} \textbf{0.18} & \color{unsat} \textbf{0.16} & \color{unsat} \textbf{0.17} & \color{unsat} \textbf{0.18}  \\
			\midrule
			\multirow{2}{*}{$\boldsymbol{b = 2}$} & \mghyper{} \cite{FinkbeinerHH18}  & \color{sat}  \textbf{0.02} & - & - & - & -\\
			& \tool{} & \color{sat} 0.17 & \color{sat}\textbf{0.18} & \color{sat} \textbf{0.19} & \color{sat} \textbf{0.24} & \color{unsat} \textbf{0.19}\\
			\bottomrule
		\end{tabular}
	\end{table}

	\begin{table}[!t]
		\caption{We check $\mathit{unsat}(n)$ where $n$ is given by the column.
			The timeout (denoted ``-'') is set to 60 seconds. }\label{tab:unsat}
		
		\vspace{1mm}
		\centering
		
		\small
		
		\def\arraystretch{1.1}
		\setlength\tabcolsep{1.7mm}

		\begin{tabular}{lcccccc}
			\toprule
			& \textbf{0} & \textbf{1} & \textbf{2} & \textbf{3} & \textbf{4} & \textbf{5} \\
			\midrule
			\mghyper{} \cite{FinkbeinerHH18}& - & - & - & - & - &  - \\
			
			\lmhyper{} \cite{BeutnerCFHK22} & \color{unsat}0.43 & \color{unsat}0.75 & \color{unsat}1.02 & \color{unsat}1.22 & \color{unsat}1.08 & \color{unsat}1.73  \\
			
			\tool{} & \color{unsat} \textbf{0.17} & \color{unsat}\textbf{0.17} & \color{unsat}\textbf{0.17} & \color{unsat}\textbf{0.19} & \color{unsat}\textbf{0.18} & \color{unsat}\textbf{0.25} \\
			\bottomrule
		\end{tabular}
		
	\end{table}

\paragraph{Unsatisfiable Formulas.}

Consider the following formula taken from \cite{BeutnerCFHK22}
\begin{align*}
	\mathit{unsat}(n) := \forall \pi_1\ldot \exists \pi_2 \ldot \exists \pi_3\ldot
	a_{\pi_3}  \land \ltlG\big(a_{\pi_1} \to \ltlN a_{\pi_2}\big) \land \ltlN^n \ltlG \neg a_{\pi_1}
\end{align*}
It states that \textbf{(1)} there exists some trace $\pi_3$ for which $a$ holds in the first step, \textbf{(2)} for every trace $\pi_1$  satisfying $a$ at some position, some trace $\pi_2$ satisfies $a$ in the next step, and \textbf{(3)} for every trace $\pi_1$, $a$ does not hold after $n$ steps. 
It is easy to see that this formula is unsatisfiable: some trace satisfies $a$ in the first position, so some trace satisfies it in the second position, etc.; eventually contradicting \textbf{(3)}. 
We check $\mathit{unsat}(n)$ in \Cref{tab:unsat}.
As \mghyper{} can never show unsatisfiability, it times out in all instances.  
\lmhyper{} can find this violation using $n$ iterative automata projections \cite{BeutnerCFHK22}, increasing its running time with larger values of $n$. 
The FOL encoding generated by \tool{} enables \texttt{Vampire} to derive unsatisfiability via pure FOL resolution without performing any automata-based reasoning.

\paragraph{Challenging Formulas for \tool{}.}

We note that current FOL solvers struggle with certain types of HyperLTL formulas, in particular those where models are \emph{infinite}.
As an example, consider the following formula:
\begin{align*}
	\forall \pi_1. \exists \pi_2 . \exists \pi_3\ldot
	a_{\pi_3} \land  (\neg a_{\pi_1}) \ltlW (a_{\pi_1} \land \ltlN \ltlG \neg a_{\pi_1})  \land \ltlG(a_{\pi_1} \to \ltlN a_{\pi_2})
\end{align*}
It states that \textbf{(1)} some trace $\pi_3$ satisfies $a$ in the first position, \textbf{(2)} $a$ should hold at most once on each trace, and \textbf{(3)} for each trace $\pi_1$  satisfying $a$ at some position, some trace $\pi_2$ satisfies $a$ in the next position. 
This formula is satisfiable, but any satisfying model contains \emph{infinitely many traces}. 
All FOL solvers we tried were unable to prove the existence of such an infinite first-order model.

\section{Evaluation on Formulas Beyond $\exists^*\forall^*$ and $\forall^*\exists^*$}\label{sec:beyond}

\begin{table}[!t]
	\caption{We check for implications between \ref{eq:gni} and \ref{eq:ni}.
		To ensure temporal safety, we use a bounded globally $\protect\ltlG^{\leq b}$ (where $b$ is given by the column). 
		The timeout (denoted ``-'') is set to 60 seconds.}\label{tab:implications}
	
	\vspace{1mm}
	\centering
	\small
	
	\def\arraystretch{1.2}
	\setlength\tabcolsep{1.7mm}
	
	\begin{tabular}{c@{\hspace{3mm}}lcccccc}
		\toprule
		\textbf{Implication} & & \textbf{1} & \textbf{2 }& \textbf{3} & \textbf{4} & \textbf{5} & \textbf{6}  \\
		\midrule
		\multirow{2}{*}{\ref{eq:gni} $\to$ \ref{eq:ni}}& \mghyper{} \cite{FinkbeinerHH18} & \color{sat} \textbf{0.02} & \color{sat} \textbf{0.02} &\color{sat}  \textbf{0.02} & \color{sat} \textbf{0.02} & \color{sat} \textbf{0.02} & \color{sat} \textbf{0.03}  \\
		& \tool{} & \color{sat} 0.27 & \color{sat} 21.3 & \color{sat} 21.5 & \color{sat} 19.5 & \color{sat} 21.5 & \color{sat} 22.0 \\
		\midrule 
		\multirow{2}{*}{\ref{eq:ni} $\to$ \ref{eq:gni}}& \mghyper{} \cite{FinkbeinerHH18} & - & - & - & - & - &  - \\
		& \tool{} & \color{sat} \textbf{0.21} & \color{sat} \textbf{4.66} & \color{sat}  \textbf{4.74} & \color{sat} \textbf{4.89} & \color{sat} \textbf{4.98} &  \color{sat} \textbf{5.05} \\
		\bottomrule
	\end{tabular}
	
\end{table}

In the previous section, we mostly evaluated on $\forall^1\exists^*$ formulas, as this is the largest fragment supported by \lmhyper{}. 
In practice, most formulas include more than one universal quantifier (see, e.g., \ref{eq:gni} or \emph{observational determinism} \cite{ZdancewicM03}), leaving \mghyper{} \cite{FinkbeinerHH18} and \tool{} as the only applicable tools.  

\paragraph{Implications Between \ref{eq:gni} and \ref{eq:ni}.}

We use \mghyper{} and \tool{} to analyze the implications between bounded versions of \ref{eq:gni} and \ref{eq:ni} in \Cref{tab:implications}.
Our analysis reveals that \ref{eq:gni} and \ref{eq:ni} are incomparable, i.e., neither formula implies the other (so the HyperLTL query is satisfiable). 
To disprove ``\ref{eq:gni} $\to$ \ref{eq:ni}'', \mghyper{} is very efficient as this implication can be disproven using a singleton model (choose a model where no trace satisfies $\ltlG \neg h$).
In contrast, disproving ``\ref{eq:ni} $\to$ \ref{eq:gni}'' requires at least two traces. 
As we have already observed in \Cref{tab:enforce}, \mghyper{} struggles with non-singleton models.

\begin{table}[!t]
	\caption{
		We compare \mghyper{} \cite{FinkbeinerHH18} and \tool{} on hand-crafted instances.
		We display the runtime in \textcolor{sat}{\textbf{green}} if the formula is satisfiable or \textcolor{unsat}{\textbf{red}} if it is unsatisfiable.
		The timeout (denoted ``-'')  is set to 60 seconds. }\label{tab:examples}
	
	\vspace{1mm}
	\centering
	\small
	
	\def\arraystretch{1.1}
	\setlength\tabcolsep{1mm}

	\begin{tabular}{lcccccc}
		\toprule
		& \textsc{Gni$\to$Ni$_+$} & \textsc{GniLeak} & \textsc{GniLeak2} & \textsc{NiLeak2} & \textsc{AnonOd} &  \textsc{AnonLeak}   \\
		\midrule
		\mghyper{} \cite{FinkbeinerHH18} & - & \color{sat}\textbf{0.03} & - & - & -& - \\
		\tool{} & \color{unsat}\textbf{0.23} & \color{sat}0.28 & \color{unsat}\textbf{0.25} & \color{unsat}\textbf{0.19} & \color{unsat}\textbf{0.34} & \color{unsat}\textbf{0.17} \\ 
		\bottomrule
	\end{tabular}
	
\end{table}

\paragraph{Hand-Crafted Examples.}

We manually design challenging instances based on \ref{eq:gni} and \ref{eq:ni}, given in \Cref{tab:examples}. 
While \Cref{tab:implications} shows that \ref{eq:gni} does not imply \ref{eq:ni}, the implication does hold (so the query is unsatisfiable) under the additional assumption $\exists \pi\ldot \ltlG (\neg h_\pi)$ (instance ``\textsc{Gni$\to$Ni$_+$}''). 
Moreover, we can use \tool{} to check if information flow properties contradict additional functional specifications (trace properties). 
For example, the functional specification $\psi_\mathit{leak} := \forall \pi. \ltlG (h_\pi \leftrightarrow o_\pi)$ requires that $h$ directly flows to the output. 
\ref{eq:gni} in conjunction with $\psi_\mathit{leak}$ is satisfiable (instance ``\textsc{GniLeak}''), but becomes unsatisfiable as soon as we require that there should be at least two different high-security inputs (instance ``\textsc{GniLeak2}''), and similarly for \ref{eq:ni} (instance ``\textsc{NiLeak2}'').
$k$-anonymity \cite{Sweene02} requires that each low-security observation can be produced by at least $k$ different high-security inputs, which is expressible as a $\forall\exists^k$ HyperLTL formula.  
\tool{} proves that $k$-anonymity (for $k = 2$) cannot be satisfied in conjunction with observational determinism \cite{ZdancewicM03} (requiring an equal output on all traces) (instance ``\textsc{AnonOd}'').
Likewise, $k$-anonymity is unsatisfiable in conjunction with $\psi_\mathit{leak}$ (instance ``\textsc{AnonLeak}'').
Note that most of these formulas are \emph{unsatisfiable}, so \tool{} is the first tool to solve them.

\section{Summary and Conclusion}

In this paper, we have presented \tool{}, a novel satisfiability checker for HyperLTL that leverages the extensive (tool) development within the FOL community for the analysis of hyperproperties. 
On the theoretical side, we study model-theoretical properties of our encoding and present a novel encoding into FOL modulo theories that is, for the first time, applicable to full HyperLTL.
Our experiments attest that  \tool{} is very effective on temporally safe hyperproperties, and occupies a useful middle ground between existing tools. 
Compared to  \eahyper{} and \lmhyper{}, \tool{} is not limited to the $\exists^*\forall^*$ and $\forall^1\exists^*$ fragments, and -- even for $\exists^*\forall^*$ formulas -- seems to handle large numbers of quantifiers better (\Cref{tab:eval-qn}).
For formulas that can be satisfied by a model consisting of a \emph{single} trace, \mghyper{} performs generally faster (\Cref{tab:random_fe}). 
However, as soon as the property requires a non-singleton model, \mghyper{} struggles. 
On the important fragment beyond $\exists^*\forall^*$ \tool{} is thus \textbf{(1)} the first tool that is efficient for formulas that require non-singleton models (\Cref{tab:enforce,tab:implications}), and \textbf{(2)} the first tool that can show unsatisfiability and thus \emph{prove} implications (\Cref{tab:unsat,tab:examples}). 

\subsubsection*{Acknowledgments.}

This work was partially supported by the European Research Council (ERC) Grant HYPER (101055412) and by the German Research Foundation (DFG) as part of TRR 248 (389792660).

\bibliographystyle{splncs04}
\bibliography{references}

\iffullversion
\appendix

\newpage

\section{HyperLTL Semantics}\label{app:semantics}

We evaluate a HyperLTL formula in the context of a set of traces $\traceset \subseteq (2^\ap)^\omega$ and a trace assignment $\Pi : \pathVars \rightharpoonup \traceset$:
\begin{align*}
	\Pi, i &\models_\traceset a_\pi &\text{iff } \quad&a \in \Pi(\pi)(i)\\
	\Pi, i&\models_\traceset \psi_1 \land \psi_2 &\text{iff } \quad &\Pi,i \models_\traceset \psi_1 \text{ and } \Pi, i \models_\traceset \psi_2\\
	\Pi, i &\models_\traceset \neg \psi &\text{iff } \quad &\Pi, i \not\models_\traceset \psi\\
	\Pi, i &\models_\traceset \ltlN \psi &\text{iff } \quad &\Pi, i+1 \models_\traceset \psi\\
	\Pi, i &\models_\traceset \psi_1 \ltlU \psi_2 &\text{iff } \quad &\exists j\ldot  j \geq i \land \Pi, j \models_\traceset \psi_2 \text{ and }  \forall k\ldot  i \leq k < j \to  \Pi, k \models_\traceset \psi_1 \span \span\\
	\Pi, i &\models_\traceset \exists \pi \ldot \varphi  &\text{iff } \quad&\exists t \in \traceset  \ldot \Pi[\pi \mapsto t], i  \models_\traceset  \varphi\\
	\Pi, i &\models_\traceset \forall \pi \ldot \varphi  &\text{iff } \quad&\forall t \in \traceset  \ldot \Pi[\pi \mapsto t], i  \models_\traceset  \varphi
\end{align*}
The atomic formula $a_\pi$ holds whenever $a$ holds in the current position $i$ on the trace bound to $\pi$.
Boolean and temporal operators are evaluated as expected, and quantification adds traces from $\traceset$ to $\Pi$.
We say a set of traces $\traceset$ satisfies $\varphi$, written $\traceset \models \varphi$, if $\{\}, 0 \models_\traceset \varphi$, where $\{\}$ denotes the trace assignment with empty domain. 
We say a formula $\varphi$ is \emph{satisfiable} if there exists a $\traceset \neq \emptyset$ with $\traceset \models \varphi$.

\section{Proof of \Cref{theo:correct1}}\label{app:proof}

In this section, we prove \Cref{theo:correct1}. 
Note that, due to our simplified encoding (modeling successor time points using a function opposed to a predicate, cf.~\Cref{rem:old}), our proof is significantly simpler than the proof for the predicate-based encoding in \cite{BeutnerCFHK22}.

\correct*
\begin{proof}
	We show both directions of the equivalence. 
	
	\noindent
	\textit{First Direction:}
	We assume that $\varphi$ is satisfiable an $\traceset \subseteq (2^\ap)^\omega$ is a set of traces with $\traceset \neq \emptyset$ and $\traceset \models \varphi$. 
	We construct the FOL interpretation $\calI_{\traceset}$ for signature $(\calS, \frakF, \frakP)$ from \Cref{sec:safe}.
	For the two sorts $\calS = \{\traceSort, \timeSort\}$ we define
	\begin{align*}
		\traceSort^{\calI_{\traceset}} := \traceset \quad\quad\quad\quad\quad\quad \timeSort^{\calI_{\traceset}} := \nat
	\end{align*}
	That is, we interpret the sort of traces as $\traceset$ and will use the natural numbers for timesteps. 
	For the function symbols we define 
	\begin{align*}
		(i_0)^{\calI_{\traceset}} := 0 \quad\quad\quad\quad (t_0)^{\calI_{\traceset}} := t \quad\quad\quad\quad	\mathit{succ}^{\calI_{\traceset}} := \lambda n : \nat \ldot n + 1
	\end{align*}
	where $t \in \traceset$ is any trace (recall that we assumed $\traceset \neq \emptyset$).
	That is, the initial time point is $0$, the fixed trace constant (used to ensure that the trace-sort is non-empty) is mapped to some arbitrary trace, and the successor function simply increments the natural number. 
	For the predicates $\{P_a\}_{a \in \ap}$ of sort $\traceSort \times \timeSort$ we define 
	\begin{align*}
		(P_a)^{\calI_{\traceset}} := \big\{ (t, i) \in \traceset \times \nat \mid a \in t(i)  \big\} \subseteq \traceSort^{\calI_{\traceset}} \times \timeSort^{\calI_{\traceset}} ,
	\end{align*}
	i.e., $P_a$ holds exactly  on those trace-timepoint pairs $(t, i)$ where $a$ holds at the $i$th position on $t$. 
	Lastly, we interpret the $\{\statePred{q}\}_{q \in Q_\psi}$ predicates of sort $\traceSort^n \times \timeSort$ as follows:
	Consider any $n$-tuple $(t_1, \ldots, t_n) \in \traceset^n$.
	We can construct the combined word $u_{t_1, \ldots, t_n} \in (2^{\ap \times \{\pi_1, \ldots, \pi_n\}})^\omega$ defined by 
	\begin{align*}
		 u_{t_1, \ldots, t_n}(i) := \bigcup_{j= 1}^n \big\{ (a, \pi_j) \mid a \in t_j(i)\big\}.
	\end{align*}
	Now assume that $t_1, \ldots, t_n$ are such that $[\pi_1 \mapsto t_1, \ldots, \pi_n \mapsto t_n] \models_\traceset \psi$.
	By assumption that $\calA_\psi$ is equivalent to $\psi$, this implies that $\calA_\psi$ has an accepting run on $u_{t_1, \ldots, t_n}$.
	We define $\rho_{t_1, \ldots, t_n} \in Q_\psi^\omega$ as any fixed accepting run on $u_{t_1, \ldots, t_n}$. 
	We can then define 
	\begin{align*}
		(\statePred{q})^{\calI_{\traceset}} := \big\{ (t_1, \ldots, t_n, i) \mid [\pi_1 \mapsto t_1, \ldots, \pi_n \mapsto t_n] \models_\traceset \psi \land \rho_{t_1, \ldots, t_n}(i) = q \big\}
	\end{align*}
	That is, we map $\statePred{q}$ to all tuples $(t_1, \ldots, t_n, i)$ where $t_1, \ldots, t_n$ satisfy $\psi$ and the unique fixed run $\rho_{t_1, \ldots, t_n}$ is in state $q$ in the $i$th step. 
	A simple induction shows that $\calI_{\traceset}$ satisfies $\theta_\varphi$. \\
	
	\noindent
	\textit{Second Direction:}
	For the second direction, assume that $\calI$ is a FOL interpretation that satisfies $\theta_\varphi$.
	Let $X = \timeSort^\calI$ be the interpretation of $\timeSort$, and define the sequence $x_0, x_1, \ldots \in X^\omega$ inductively by 
	\begin{align*}
		x_0 &:= (i_0)^\calI \quad\quad\quad\quad\quad\quad x_{i + 1} := \mathit{succ}^{\calI} (x_i).
	\end{align*}
	That is, we use $\calI$'s interpretation of $i_0 \in \frakF$ for the initial time point, and then apply $\calI$'s interpretation of $\mathit{succ} \in \frakF$. 
	Intuitively, $x_i$ is the $i$-fold application of $\mathit{succ}^{\calI}$ on $(i_0)^\calI$, which, in our encoding, we can view as the $i$th timepoint.
	Note that the $x_i$s may not be distinct, i.e., the time might be cyclic. 
	For each $y \in \traceSort^\calI$, we define a trace $t_y \in (2^\ap)^\omega$ pointwise as follows: 
	For each $i \in \nat$, we define the $i$th position of $t_y$ (i.e., $t_y(i) \subseteq 2^\ap$) as follows
	\begin{align*}
		t_y(i) := \big\{a \in \ap \mid (y, x_i) \in (P_a)^\calI \big\}.
	\end{align*}
	That is, $a \in \ap$ holds in the $i$th step on $t_y$ iff $(y, x_i)$ is in $\calI$'s interpretation of $P_a$.
	We define $\traceset := \{t_y \mid y \in \traceSort^\calI\}$.
	Note that $\traceset \neq \emptyset$ as $\traceSort^\calI \neq \emptyset$ (as $(t_0)^\calI \in \traceSort^\calI$).
	A simple induction shows that $\traceset \models \varphi$ as required. 
	\qed
\end{proof}

\section{Proof of \Cref{prop:finite}}\label{app:finite}

\finite*
\begin{proof}
	The ``if'' direction follows directly from the construction in the proof of \Cref{theo:correct1}:
	Assume $\calI$ is a finite interpretation of $\theta_\varphi$.
	In the construction of $\traceset$ (the second direction in the proof of \Cref{theo:correct1}), we add one trace $t_y$ for each $y \in \traceSort^\calI$.
	As $\calI$ is a finite, so is $\traceSort^\calI$ and thus $\traceset$. 
	
	The ``only if'' direction is more challenging. 
	Note that the construction in \Cref{theo:correct1} uses $\nat$ as the interpretation of $\timeSort$ and is thus always infinite. 
	Assume that $\traceset \subseteq (2^\ap)^\omega$ is such that $\traceset \models \varphi$ and $|\traceset| < \infty$. 
	As $\traceset$ is finite, we can assume, w.l.o.g., that $\traceset$ consists of \emph{lasso-shaped} traces, i.e., we can write each $t \in \traceset$ as $t = t^\mathit{stem} (t^\mathit{loop})^\omega$ for some $t^\mathit{stem}, t^\mathit{loop} \in (2^\ap)^+$.
	Note that we do not claim that all finite models are lasso shaped, but rather that if there exists a finite model there also exists a finite model consisting only of lasso-shaped traces. 
	Now consider any $n$-tuple $(t_1, \ldots, t_n) \in \traceset^n$.
	We can construct the combined word $u_{t_1, \ldots, t_n} \in (2^{\ap \times \{\pi_1, \ldots, \pi_n\}})^\omega$ defined by 
	\begin{align*}
		 u_{t_1, \ldots, t_n}(i) := \bigcup_{j= 1}^n \big\{ (a, \pi_j) \mid a \in t_j(i) \big\}.
	\end{align*}
	Now assume that $t_1, \ldots, t_n$ satisfy $\psi$, i.e., $[\pi_1 \mapsto t_1, \ldots, \pi_n \mapsto t_n] \models_\traceset \psi$. 
	By assumption that $\calA_\psi$ is equivalent to $\psi$, this implies that $\calA_\psi$ has an accepting run on $u_{t_1, \ldots, t_n}$.
	Moreover, as $t_1, \ldots, t_n$ are all lasso-shaped, $u_{t_1, \ldots, t_n}$ is lasso-shaped as well, so $\calA_\psi$ has a lasso-shaped run on $u_{t_1, \ldots, t_n}$.
	For any $t_1, \ldots, t_n$ with $[\pi_1 \mapsto t_1, \ldots, \pi_n \mapsto t_n] \models_\traceset \psi$, we can therefore define $\rho_{t_1, \ldots, t_n} \in Q_\psi^\omega$ as some fixed accepting \emph{lasso-shaped} run on $u_{t_1, \ldots, t_n}$. 
	Assume $\rho_{t_1, \ldots, t_n} = \rho^\mathit{stem}_{t_1, \ldots, t_n}( \rho^\mathit{loop}_{t_1, \ldots, t_n})^\omega$ for some $\rho^\mathit{stem}_{t_1, \ldots, t_n}, \rho^\mathit{loop}_{t_1, \ldots, t_n} \in Q_\psi^+$.
	
	We now consider all lassos $\{t^\mathit{stem} (t^\mathit{loop})^\omega \mid t \in \traceset\}$ and $\{\rho^\mathit{stem}_{t_1, \ldots, t_n}( \rho^\mathit{loop}_{t_1, \ldots, t_n})^\omega \mid t_1, \ldots, t_n \in \traceset \land [\pi_1 \mapsto t_1, \ldots, \pi_n \mapsto t_n] \models_\traceset \psi \}$. 
	We can assume, w.l.o.g., that the stem of all those lassos have the same length (we can always extend the stem by moving a prefix of the unrolled lasso to the stem). 
	Let $M_\mathit{stem} \in \nat$ be the common length of all stems. 
	To model the loop-part of all lassos, we need some common multiple.
	We define
	\begin{align*}
		M_\mathit{loop} := \mathit{lcm}\bigg( &\Big\{ | t^\mathit{loop} | \mid t \in \traceset \Big\} \, \cup \\
		&\Big\{ |\rho^\mathit{loop}_{t_1, \ldots, t_n}|  \mid t_1, \ldots, t_n \in \traceset \land [\pi_1 \mapsto t_1, \ldots, \pi_n \mapsto t_n] \models_\traceset \psi \Big\} \bigg),
	\end{align*}	
	where $\mathit{lcm}$ denotes the \emph{least common multiple}.
	That is, $M_\mathit{loop}$ is the least common multiple of the loop lengths of all lassos. 
	The idea is that we interpret the time as the \emph{finite} set $\{0, \ldots, M_\mathit{stem} + M_\mathit{loop}-1\}$.
	For any $i \in \nat$, we can define the finite index $c(i) \in \{0, \ldots, M_\mathit{stem} + M_\mathit{loop}-1\}$ by computing modulo $M_\mathit{loop}$:
	\begin{align*}
		\mathit{c}(i) := \begin{cases}
			\begin{aligned}
				&i \quad &&\text{if } i < M_\mathit{stem}\\
				&M_\mathit{stem} + \big( (i - M_\mathit{stem}) \texttt{ mod } M_\mathit{loop}\big)  &&\text{otherwise}
			\end{aligned}
		\end{cases}
	\end{align*}
	For example, if $M_\mathit{stem} = 2, M_\mathit{loop} = 3$, we have
	\begin{center}
		\def\arraystretch{1.1}
		\setlength\tabcolsep{1.5mm}
		\begin{tabular}{ccccccccccc}
			\toprule
			$i$ & 0 & 1 & 2 & 3  & 4  & 5  & 6 & 7 & 8 & $\cdots$  \\
			\midrule
			$c(i)$ & 0 & 1 & 2 & 3  & 4 & 2 & 3 & 4 & 2 & $\cdots$ \\
			\bottomrule
		\end{tabular}
	\end{center}
	
	\noindent
	Now -- as all stems have the same length ($M_\mathit{stem}$) and $M_\mathit{loop}$ is a common multiple of all loop-lengths -- we can describe all traces and runs within this finite set. 
	Concretely, we have the following: 
	For any $i \in \nat$ and any trace $t \in \traceset$, we have $t(i) = t\big(c(i) \big)$. 
	Likewise, for all $t_1, \ldots, t_n$ with $[\pi_1 \mapsto t_1, \ldots, \pi_n \mapsto t_n] \models_\traceset \psi$, we have $\rho_{t_1, \ldots, t_n}(i) = \rho_{t_1, \ldots, t_n}\big(c(i)\big)$.
	Call this observation \textbf{(1)}.
	
	With this intuition, we can define our interpretation $\calI_{\traceset}$ as follows:
	For the two sorts $\calS = \{\traceSort, \timeSort\}$ we define
	\begin{align*}
		\traceSort^{\calI_{\traceset}} := \traceset \quad\quad\quad\quad\quad\quad \timeSort^{\calI_{\traceset}} := \{0, \ldots, M_\mathit{stem} + M_\mathit{loop}-1\}
	\end{align*}
	For the constants (nullary functions), we define 
	\begin{align*}
		(i_0)^{\calI_{\traceset}} := 0 \quad\quad\quad (t_0)^{\calI_{\traceset}} := t
	\end{align*}
	where $t \in \traceset$ is any trace (recall that we assumed $\traceset \neq \emptyset$).
	For the successor function we define 
	\begin{align*}
		\mathit{succ}^{\calI_{\traceset}} := \lambda n : \{0, \ldots, M_\mathit{stem} + M_\mathit{loop}-1\} \ldot\begin{cases}
			\begin{aligned}
				&n + 1\quad &&\text{if } n < M_\mathit{stem} + M_\mathit{loop}-1\\
				&M_{\mathit{stem}}\quad &&\text{if } n = M_\mathit{stem} + M_\mathit{loop}-1
			\end{aligned}
		\end{cases}
	\end{align*}
	That is, $\mathit{succ}^{\calI_{\traceset}}$ loops through $\{0, \ldots, M_\mathit{stem} + M_\mathit{loop}-1\}$; restarting in $M_\mathit{stem}$. The $i$-fold application of $\mathit{succ}^{\calI_{\traceset}}$ to $(i_0)^{\calI_{\traceset}} = 0$, thus yields $c(i)$. 
	
	For the predicates $\{P_a\}_{a \in \ap}$ we define 
	\begin{align*}
		(P_a)^{\calI_{\traceset}} := \big\{ (t, i) \mid a \in t(i)  \big\}, 
	\end{align*}
	i.e., $P_a$ holds exactly  on those trace-timepoint pairs $(t, i)$ where $a$ holds at the $i$th position on $t$ (for $i \in \{0, \ldots, M_\mathit{stem} + M_\mathit{loop}-1\}$).
	Note that by \textbf{(1)}, the first $M_\mathit{stem} + M_\mathit{loop}-1$ positions suffice to fully recover $t$.
	Lastly, we interpret $\statePred{q}$ as
	\begin{align*}
		(\statePred{q})^{\calI_{\traceset}} := \Big\{ (t_1, \ldots, t_n, i) \mid [\pi_1 \mapsto t_1, \ldots, \pi_n \mapsto t_n] \models_\traceset \psi \land \rho_{t_1, \ldots, t_n}(i) = q \Big\}.
	\end{align*}
	That is, we map $\statePred{q}$ to all tuples $(t_1, \ldots, t_n, i)$ where $t_1, \ldots, t_n$ satisfy $\psi$ and the unique fixed run $\rho_{t_1, \ldots, t_n}$ is in state $q$ in the $i$th step. 
	Again, by \textbf{(1)} this already defines the entire on the cyclic time-step set.
	
	A simple induction shows that $\calI_{\traceset}$ satisfies $\theta_\varphi$.
	Note that $\calI^\traceset$ is finite as required. 
	\qed
\end{proof}

\section{FOL Encoding For Full HyperLTL}\label{app:full}

We modify our encoding from \Cref{sec:safe}, to support full HyperLTL. 

\subsection{Preliminaries}
 
\paragraph{Büchi Automata.}

An nondeterministic Büchi automaton (NBA) over some alphabet $\Sigma$ is a tuple $\calA = (Q, Q_0, \delta, F)$ where $Q$ is a finite set of states, $Q_0 \subseteq Q$ is a set of initial states, $\delta \subseteq Q \times \Sigma \times Q$ is a transition relation, and $F \subseteq Q$ is a set of accepting sets. 
A run of $\calA$ on a word $u \in \Sigma^\omega$ is an infinite sequence $\rho \in Q^\omega$ such that $\rho(0) \in Q_0$ and for every $i \in \nat$, $(\rho(i), u(i), \rho(i+1)) \in \delta$. 
The run $\rho$ is accepting if it visits states in $F$ \emph{infinitely} many times.

\paragraph{First-Order Logic Modulo Theories.}

In pure FOL, all function and predicate symbols are uninterpreted, i.e., an interpretation can fix all function and predicate symbols. 
In our modified encoding, we use FOL modulo a fixed backgorund theory. 
Intuitively, a theory fixes some sorts, functions, and predicates. 
For example, in FOL modulo linear integer aromatic (LIA), we fix a dedicated sort $\intSet \in \calS$, function symbols $+, 0, 1 \in \frakF$, and predicate symbol $\leq \in \frakP$; all with the expected fixed interpretation. 

\begin{figure}[!t]
	\begin{align}
		&\quant_1 x_1 : \traceSort{} \ldots \quant_n x_n  : \traceSort{} \ldot \label{eq:full-quant} \\
		&\quad\quad\bigg(\bigvee_{q \in Q_{\psi, 0}} \statePred{q}(x_1, \ldots, x_n, 0) \bigg) \, \land \label{eq:full-init} \\
		&\quad\quad\Bigg(\forall i : \intSet{}\ldot \bigwedge_{q \in Q_\psi}  \statePred{q}(x_1, \ldots, x_n, i) \to \label{eq:full-premise} \\
		&\quad\quad\quad\quad\bigg(\bigvee_{(q, \sigma, q') \in \delta_\psi} \bigg[ \bigwedge_{a_{\pi_j} \in \sigma} P_a(x_j, i) \land \bigwedge_{a_{\pi_j} \not\in \sigma} \neg P_a(x_j, i) \, \land \label{eq:full-cons1} \\
		&\quad\quad\quad\quad\quad\quad\quad\quad\quad\quad\quad\statePred{q'}\big(x_1, \ldots, x_n,i + 1\big) \bigg] \bigg)\Bigg) \, \land \label{eq:full-cons2} \\
		&\quad\quad\Big(\forall i : \intSet{}\ldot \exists i' : \intSet{}\ldot i < i' \land \bigwedge_{q \in Q_\psi \setminus F_\psi} \neg \statePred{q}(x_1, \ldots, x_n, i')\Big) \label{eq:full-safe}
	\end{align}

	\caption{We depict the FOL module LIA.}\label{fig:enc-full}
\end{figure}

\begin{table}[!t]
	
	\caption{We comapre \eahyper{} \cite{FinkbeinerHS17} (using \texttt{PLTL} \cite{Schwendimann98} as the LTL satisfiability solver), with \tool{} (using \texttt{Vampire} \cite{KovacsV13} and \texttt{Paradox} \cite{claessen2003new} as FOL solvers) on randomly generated $\exists^n\forall^m$ formulas (where $n$ gives the row and $m$ gives the column). For each combination $(n, m)$ we generate 10 random formulas.
		Each cell has the form $p, t$ where $p$ denotes the percentage solved (between $0$ and $1$), and $t$ is the average time needed on the solved instances (in seconds). The timeout is set to 10 seconds.}\label{tab:ea-random}
	\vspace{2mm}
	\centering
	\small
	
	\def\arraystretch{1.1}
	\setlength\tabcolsep{1.2mm}
	
	\scalebox{0.65}{
		\begin{tabular}{l@{\hspace{4mm}}llllllllll}
			\toprule
			&&\textbf{1} & \textbf{2} & \textbf{3} & \textbf{4} & \textbf{5} & \textbf{6} & \textbf{7} & \textbf{8} & \textbf{9} \\
			\midrule
			\multirow{3}{*}{\textbf{1}}&\eahyper{} \cite{FinkbeinerHS17} & 1.00, 0.01 & 1.00, 0.01 & 1.00, 0.01 & 1.00, 0.01 & 1.00, 0.01 & 1.00, 0.01 & 1.00, 0.01 & 1.00, 0.01 & 1.00, 0.01   \\
			& \tool{} (\texttt{Vampire}) & 1.00, 0.19 & 1.00, 0.17 & 0.90, 0.18 & 0.80, 0.18 & 0.90, 0.18 & 1.00, 0.18 & 1.00, 0.19 & 1.00, 0.19 & 1.00, 0.19   \\
			&\tool{} (\texttt{Paradox}) & 1.00, 0.20 & 1.00, 0.20 & 1.00, 0.20 & 1.00, 0.20 & 1.00, 0.21 & 1.00, 0.20 & 1.00, 0.19 & 1.00, 0.20 & 1.00, 0.19  \\
			\midrule
			\multirow{3}{*}{\textbf{2}}&\eahyper{} \cite{FinkbeinerHS17} & 1.00, 0.01 & 1.00, 0.01 & 1.00, 0.01 & 1.00, 0.01 & 1.00, 0.01 & 1.00, 0.01 & 1.00, 0.01 & 1.00, 0.01 & 1.00, 0.02    \\
			& \tool{} (\texttt{Vampire}) & 0.80, 0.18 & 0.90, 0.17 & 0.90, 0.21 & 1.00, 0.20 & 1.00, 0.19 & 1.00, 0.19 & 1.00, 0.18 & 1.00, 0.18 & 0.90, 0.19   \\
			&\tool{} (\texttt{Paradox}) & 1.00, 0.21 & 1.00, 0.21 & 1.00, 0.20 & 1.00, 0.19 & 1.00, 0.19 & 1.00, 0.19 & 1.00, 0.19 & 1.00, 0.19 & 1.00, 0.20  \\
			\midrule
			\multirow{3}{*}{\textbf{3}}&\eahyper{} \cite{FinkbeinerHS17} & 1.00, 0.01 & 1.00, 0.01 & 1.00, 0.01 & 1.00, 0.01 & 1.00, 0.01 & 1.00, 0.02 & 1.00, 0.05 & 1.00, 0.07 & 1.00, 0.34   \\
			& \tool{} (\texttt{Vampire}) & 1.00, 0.18 & 1.00, 0.18 & 0.90, 0.18 & 1.00, 0.19 & 0.90, 0.20 & 1.00, 0.20 & 0.90, 0.19 & 1.00, 0.19 & 0.90, 0.18   \\
			&\tool{} (\texttt{Paradox}) & 1.00, 0.20 & 1.00, 0.20 & 0.90, 0.19 & 1.00, 0.19 & 1.00, 0.19 & 1.00, 0.20 & 0.90, 0.20 & 1.00, 0.20 & 1.00, 0.20  \\
			\midrule
			\multirow{3}{*}{\textbf{4}}&\eahyper{} \cite{FinkbeinerHS17} & 1.00, 0.01 & 1.00, 0.01 & 1.00, 0.01 & 1.00, 0.01 & 1.00, 0.02 & 1.00, 0.05 & 1.00, 0.36 & 0.00, - & 0.00, -   \\
			& \tool{} (\texttt{Vampire}) & 1.00, 0.19 & 1.00, 0.69 & 1.00, 0.18 & 1.00, 0.18 & 1.00, 0.18 & 1.00, 0.18 & 0.90, 0.18 & 1.00, 0.18 & 0.90, 0.18   \\
			&\tool{} (\texttt{Paradox}) & 1.00, 0.19 & 1.00, 0.19 & 1.00, 0.19 & 1.00, 0.19 & 1.00, 0.19 & 1.00, 0.19 & 1.00, 0.20 & 1.00, 0.19 & 1.00, 0.19  \\
			\midrule
			\multirow{3}{*}{\textbf{5}}&\eahyper{} \cite{FinkbeinerHS17} & 1.00, 0.01 & 1.00, 0.01 & 1.00, 0.01 & 1.00, 0.02 & 1.00, 0.04 & 1.00, 0.25 & 0.00, - & 0.00, - & 0.00, -   \\
			& \tool{} (\texttt{Vampire}) & 1.00, 0.18 & 1.00, 0.18 & 0.90, 0.18 & 1.00, 0.17 & 1.00, 0.18 & 0.90, 0.18 & 0.90, 0.18 & 1.00, 0.18 & 0.90, 0.18  \\
			&\tool{} (\texttt{Paradox}) & 1.00, 0.20 & 1.00, 0.19 & 1.00, 0.20 & 1.00, 0.19 & 1.00, 0.19 & 1.00, 0.20 & 1.00, 0.19 & 1.00, 0.19 & 1.00, 0.19  \\
			\midrule
			\multirow{3}{*}{\textbf{6}}&\eahyper{} \cite{FinkbeinerHS17} & 1.00, 0.01 & 1.00, 0.01 & 1.00, 0.01 & 1.00, 0.03 & 1.00, 0.18 & 0.00, - & 0.00, - & 0.00, - & 0.00, -   \\
			& \tool{} (\texttt{Vampire}) & 0.90, 0.17 & 0.90, 0.18 & 1.00, 0.18 & 1.00, 0.19 & 0.90, 0.18 & 0.90, 0.18 & 0.90, 0.19 & 0.70, 0.18 & 0.90, 0.19  \\
			&\tool{} (\texttt{Paradox}) & 1.00, 0.81 & 1.00, 0.20 & 1.00, 0.20 & 1.00, 0.19 & 1.00, 0.19 & 1.00, 0.20 & 1.00, 0.20 & 1.00, 0.20 & 1.00, 0.20  \\
			\midrule
			\multirow{3}{*}{\textbf{7}}&\eahyper{} \cite{FinkbeinerHS17} & 1.00, 0.01 & 1.00, 0.01 & 1.00, 0.02 & 1.00, 0.03 & 0.90, 0.29 & 0.00, - & 0.00, - & 0.00, - & 0.00, -  \\
			& \tool{} (\texttt{Vampire}) & 1.00, 0.18 & 1.00, 0.17 & 0.90, 0.18 & 1.00, 0.18 & 0.90, 0.18 & 1.00, 0.17 & 0.90, 0.17 & 1.00, 0.17 & 1.00, 0.18  \\
			&\tool{} (\texttt{Paradox}) & 1.00, 0.18 & 1.00, 0.17 & 0.90, 0.18 & 1.00, 0.18 & 0.90, 0.18 & 1.00, 0.17 & 0.90, 0.17 & 1.00, 0.17 & 1.00, 0.18 \\
			\midrule
			\multirow{3}{*}{\textbf{8}}&\eahyper{} \cite{FinkbeinerHS17} & 1.00, 0.01 & 1.00, 0.01 & 1.00, 0.02 & 1.00, 0.09 & 0.00, - & 0.00, -& 0.00, - & 0.00, - & 0.00, -  \\
			& \tool{} (\texttt{Vampire}) & 0.80, 0.18 & 1.00, 0.19 & 1.00, 0.18 & 0.90, 0.18 & 1.00, 0.18 & 0.80, 0.18 & 1.00, 0.19 & 1.00, 0.18 & 1.00, 0.18  \\
			&\tool{} (\texttt{Paradox}) & 1.00, 0.20 & 1.00, 0.19 & 1.00, 0.20 & 1.00, 0.19 & 1.00, 0.19 & 1.00, 0.19 & 1.00, 0.19 & 1.00, 0.19 & 1.00, 0.19 \\
			\midrule
			\multirow{3}{*}{\textbf{9}}&\eahyper{} \cite{FinkbeinerHS17} & 1.00, 0.01 & 1.00, 0.01 & 1.00, 0.02 & 1.00, 0.08 & 0.00, - & 0.00,- & 0.00, - & 0.00, - & 0.00, -  \\
			& \tool{} (\texttt{Vampire}) & 1.00, 0.17 & 0.90, 0.18 & 0.90, 0.18 & 1.00, 0.18 & 0.90, 0.18 & 0.90, 0.19 & 1.00, 0.18 & 1.00, 0.19 & 1.00, 0.18  \\
			&\tool{} (\texttt{Paradox}) & 1.00, 0.19 & 1.00, 0.19 & 1.00, 0.20 & 1.00, 0.20 & 1.00, 0.20 & 1.00, 0.20 & 1.00, 0.20 & 1.00, 0.19 & 1.00, 0.38 \\
			\bottomrule
	\end{tabular}}
\end{table}

\subsection{FOL Encoding}

Let $\varphi = \quant_1 \pi_1\ldots \quant_n \pi_n\ldot \psi$ be an \emph{arbitrary} HyperLTL formula, and let $\calA_\psi = (Q_\psi, Q_{\psi, 0}, \delta_\psi, F_\psi)$ be an NBA over alphabet $2^{\ap \times \{\pi_1, \ldots, \pi_n\}}$ that is equivalent to $\psi$. 
We use a FOL signature $(\calS, \frakF, \frakP)$ where \textbf{(1)} $\calS = \{\traceSort{}, \intSet\}$, \textbf{(2)} $\frakF = \{0 : \intSet, 1 : \intSet, + : \intSet \times \intSet \to \intSet, t_0 : \traceSort\}$, and \textbf{(3)} $\frakP$ includes the following: a predicate $< : \intSet \times \intSet$, for each $a \in \ap$ a predicate $P_a :  \traceSort{} \times \intSet$, and for each $q \in Q_\psi$ a predicate $\statePred{q} : (\times_{i=1}^n \traceSort{})\times \intSet$.
Note that $\intSet, 0, 1, +$ and $<$ have a fixed interpretation in LIA.

We define $\theta_\varphi$ as the FOL formula over $(\calS, \frakF, \frakP)$ given in \Cref{fig:enc-full}.
The main idea of our encoding is similar to the one in \Cref{sec:safe}.
The main difference is that \emph{time is now interpreted} (as $\intSet$), so we can directly compute the next timepoint (using $+ 1$) and we can compare two timepoints using $<$. 
We again mimic the quantification in $\varphi$ (\ref{eq:full-quant}) and demand that we start in some initial state at timepoint $0$ (\ref{eq:full-init}).
For each transition, we can directly access the \emph{next} timepoint using $+ 1$ (\ref{eq:full-cons2}).
Finally, using $<$ we can express the acceptance of the resulting run in the more expressive Büchi condition: we require that for every timepoint, some later timepoints only visits accepting states, i.e., does not visit a non-accepting state (\ref{eq:full-safe}).

\begin{remark}
	Note that our encoding ensures that $\statePred{q}(x_1, \ldots, x_n, i)$ encodes that a run of $\calA_\psi$ on the traces $x_1, \ldots, x_n$ \emph{can} be in state $q$ at timepoint $i$. 
	In particular, the same combination of $n$ traces and timepoint can be in multiple automaton states at the same time. 
	In our encoding it therefore does not suffice to requires that ``infinitely often we can be in accepting state'' but rather we require that ``infinitely often we cannot be in a non-accepting state, i.e., all possible automaton states are accepting''.
	\demo
\end{remark}

Similar to \Cref{theo:correct1}, we can easily show that our encoding is correct.

\begin{theorem}\label{theo:correct2}
	A HyperLTL formula $\varphi$ is satisfiable if and only if there exists a FOL interpretation $\calI$ (modulo LIA) that satisfies $\theta_\varphi$.
\end{theorem}

\section{Additional Material for \Cref{sec:ef}}\label{app:eahyper}

In \Cref{tab:ea-random}, we compare \eahyper{} and \tool{} on randomly generated $\exists^n\forall^m$ formulas. 
As in \Cref{sec:ef}, we observe that for large numbers of quantifiers \eahyper{} times out, whereas \tool{} can solve almost all examples.

\fi

\end{document}